\documentclass[12pt]{article}

\usepackage{authblk}
\usepackage[margin=1in]{geometry}
\usepackage{amsmath,amssymb,amsthm}
\usepackage[usenames,dvipsnames,svgnames,table]{xcolor}

\usepackage{graphicx}
\graphicspath{{figures/}} 

\usepackage{hyperref}
\hypersetup{
	colorlinks=true,
	linkcolor=blue,
	filecolor=magenta,      
	urlcolor=cyan,
}

\usepackage{tabularx}
\usepackage{times}

\newtheorem{theorem}{Theorem}[section]
\newtheorem{lemma}[theorem]{Lemma}
\newtheorem{proposition}[theorem]{Proposition}
\newtheorem{definition}[theorem]{Definition}
\newtheorem{corollary}[theorem]{Corollary}
\newtheorem{remark}[theorem]{Remark}

\newcommand{\R}{\mathbb{R}}

\newcommand{\op}[1]{\operatorname{#1}}
\newcommand{\norm}[1]{\left\lVert#1\right\rVert}
\newcommand{\maxnorm}[1]{\norm{#1}_{\max}}
\newcommand{\dgm}{\op{dgm}}
\newcommand{\dis}{\op{dis}}

\newcommand{\D}{\op{d_N}}
\newcommand{\dB}{\op{d_B}}
\newcommand{\dGH}{\op{d_{GH}}}
\newcommand{\VR}[2]{\op{VR}(#1,#2)}
\newcommand{\sig}[2]{\sigma_{#1}^{#2}}
\newcommand{\JL}{\mathrm{JL}}
\newcommand{\diam}{\op{diam}}

\newcommand\restr[2]{{
  \left.\kern-\nulldelimiterspace 
  #1 
  \vphantom{\big|} 
  \right|_{#2} 
  }}
  
\DeclareMathOperator*{\argmin}{arg\,min}

\usepackage[normalem]{ulem}
\definecolor{addblue}{rgb}{0.1,0,0.8}
\newcommand{\add}[1]{\textcolor{addblue}{#1}}

\definecolor{darkgrn}{rgb}{0, 0.75, 0}

\date{}

\everymath{\displaystyle}

\begin{document}

\title{{\bfseries A Normalized Bottleneck Distance on Persistence Diagrams and Homology Preservation under Dimension Reduction}}

\author[1]{Nathan H.~May}
\author[1]{Bala Krishnamoorthy\thanks{kbala@wsu.edu}}
\author[1]{Patrick Gambill}
\affil[1]{Department of Mathematics and Statistics, Washington State University, USA}


\maketitle

\begin{abstract}
  Persistence diagrams (PDs) are used as signatures of point cloud data.
  Two clouds of points can be compared using the bottleneck distance $\dB$ between their PDs.
  A potential drawback of this pipeline is that point clouds sampled from topologically similar manifolds can have arbitrarily large $\dB$ when there is a large scaling between them.
  This situation is typical in dimension reduction frameworks.
  
  We define, and study properties of, a new scale-invariant distance between PDs termed normalized bottleneck distance, $\D$.
  In defining $\D$, we develop a broader framework called metric decomposition for comparing finite metric spaces of equal cardinality with a bijection.
  We utilize metric decomposition to prove a stability result for $\D$ by deriving an explicit bound on the distortion of the bijective map.
  We then study two popular dimension reduction techniques, Johnson-Lindenstrauss (JL) projections and metric multidimensional scaling (mMDS), and a third class of general biLipschitz mappings.
  We provide new bounds on how well these dimension reduction techniques preserve homology with respect to $\D$.
  For a JL map $f:X \to f(X)$, we show that $\D(\dgm(X),\dgm(f(X))) < \epsilon$ where $\dgm(X)$ is the Vietoris-Rips PD of $X$, and pairwise distances are preserved by $f$ up to the tolerance  $0 < \epsilon < 1$.
  For mMDS, we present new bounds for $\dB$ and $\D$ between PDs of $X$ and its projection in terms of the eigenvalues of the covariance matrix.
  And for $k$-biLipschitz maps, we show that $\D$ is bounded by the product of $(k^2-1)/k$ and the ratio of diameters of $X$ and $f(X)$.
  Finally, we use computational experiments to demonstrate the increased effectiveness of using the normalized bottleneck distance for clustering sets of point clouds sampled from different shapes.

  \hspace*{0.7in} Johnson-Lindenstrauss projection, metric multidimensional scaling.
\end{abstract}

\clearpage
\section{Introduction}
Persistent homology has matured into a widely used and powerful tool in topological data analysis (TDA) \cite{EdMo2013}.
A typical TDA pipeline starts with point cloud data (PCD) assumed to be sampled from a manifold along with a specified distance metric.
The Vietoris-Rips (VR) persistence diagram of such a data set allows us to observe ``holes'' in data, which could have implications for the application generating the data \cite{Ca2009,CoSh2009}.
Going one step further, we can compare two different data sets by directly comparing the bottleneck distance between their persistence diagrams (PDs) \cite[Chap.~5]{ChdeSGlOu2016}.
The bottleneck distance can be computed efficiently, and also satisfies standard notions of stability \cite{ChdeSOu2014}.

A drawback of this TDA pipeline for comparing data sets is that PCDs sampled from homeomorphic manifolds can have an arbitrarily large bottleneck distance between their PDs when there is a large degree of scaling.
This situation is illustrated in Figure \ref{fig:SclProbIllst} on two pairs of PCDs where all four data sets are sampled from noisy circles except that there is is a large degree of scaling between the second pair of PCDs.
At a first look, all four PCDs look quite similar---each represent points sampled from a noisy circle.
One has to pay close attention to the axes scales in the second blue PCD (in the second row) to notice the difference in scales.

\begin{figure}[ht!]
  \includegraphics[height=1.18in]{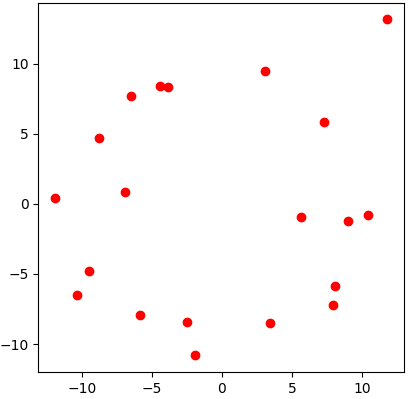}
  \includegraphics[height=1.18in]{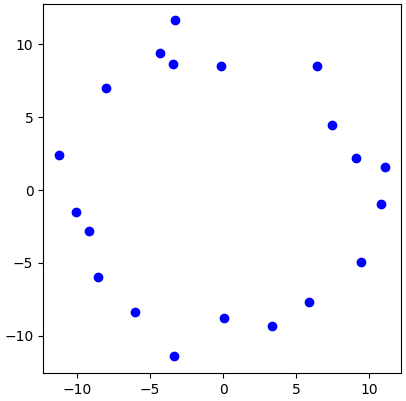}
  \includegraphics[height=1.18in]{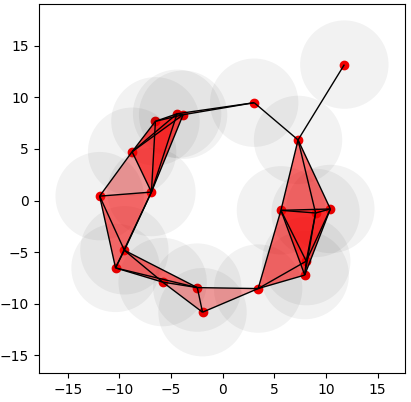}
  \includegraphics[height=1.15in]{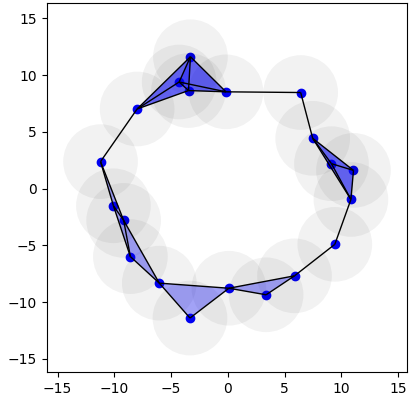}
  \hfill
  \includegraphics[height=1.1in]{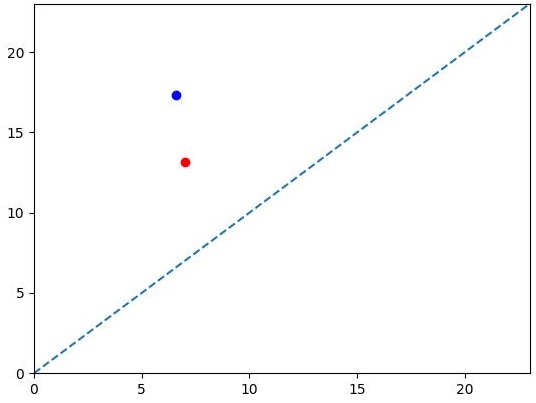}\\
  \hspace*{-0.1in}
  \includegraphics[height=1.1in]{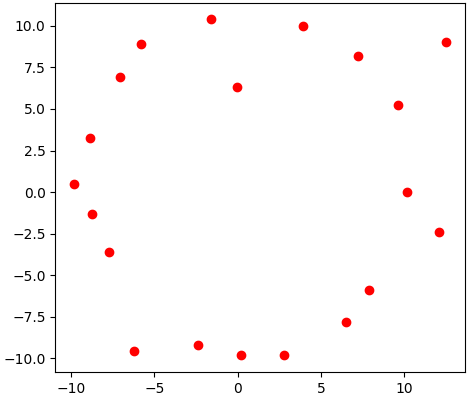}
  \includegraphics[height=1.1in]{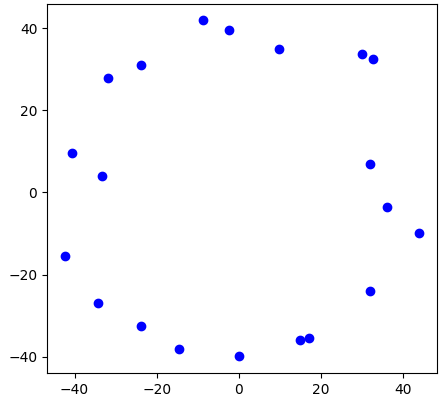}
  \includegraphics[height=1.1in]{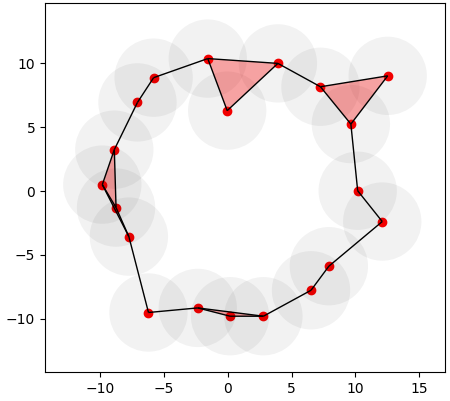}
  \includegraphics[height=1.1in]{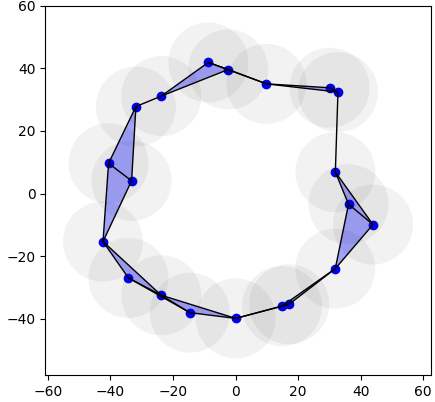}
  \hfill
  \includegraphics[height=1.1in]{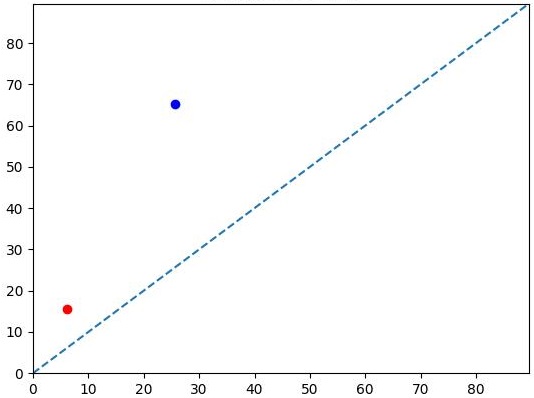}
  \caption{  \label{fig:SclProbIllst}
    Two pairs of point clouds, one without scaling (top row) and another with a large amount of scaling (bottom row), and their 1D persistence diagrams ($H_1$).
    The red and blue point clouds in the top row (first two figures) are both sampled from noisy circles in the $[-10,10] \times [-10,10]$ box.
    The third and fourth figures show the births of the hole in each case at diameters of around $7$.
    The first (red) point cloud in the bottom row is similar to the ones in the top row, but the second (blue) point cloud is sampled from a noisy circle in the $[-80,80] \times [-80,80]$ box.
    Its hole feature is born at a diameter of around $26$.
    The bottleneck distance between the pair of PDs is small in the top row, but can be quite large depending on the degree of scaling of the second (blue) point cloud in the bottom row.
  }
\end{figure}

Such scalings may have to be handled explicitly in certain applications,
e.g., in manufacturing of machine parts whose scales determine their actual sizes or in astronomy when comparing orbits of bodies around the sun where the scales are linked to their periods of revolution.
On the other hand, such scalings often show up as a natural result of algorithmic pipelines that transforms the first PCD into the second.
In fact, this is a typical scenario when the second PCD is produced by applying dimension reduction (DR) to the first PCD.
For a direct illustration, see Figure \ref{fig:Hom:UMAP_Saddle} (in Page \pageref{fig:Hom:UMAP_Saddle}) to observe how UMAP \cite{McHeMe2020}, a popular nonlinear DR method applied under default settings, shrinks a 3D loop forming a saddle boundary in a $200 \times 200 \times 100$ box to nearly a unit circle in the 2D plane.
While the topology of the saddle boundary is clearly preserved by the UMAP reduction by capturing its loop structure, standard bottleneck distance between the PDs of the input and the output representations cannot be constrained by any reasonably small bound.

We consider the following motivating questions:
Can we define a \emph{scale-invariant} bottleneck distance between PDs?
Can we derive tighter stability bounds on this distance than ones known for the standard bottleneck distance?
And can we use this new distance and associated stability bounds to capture how well standard DR techniques may preserve homology?


\subsection{Our Contributions}

We define a new scale-invariant distance between persistence diagrams termed \emph{normalized bottleneck distance}, $\D$, and study its properties (see Definition \ref{defn:Hom:dN}).
On top of being scale-invariant (see Theorem \ref{thm:Hom:dN_scaleIn}), $\D$ is a pseudometric for persistence diagrams, and also has the computational advantage of not needing to directly compute an optimal scaling/dilation factor.
  In defining $\D$, we also develop a broader framework called metric decomposition for comparing finite metric spaces of equal cardinality that also have a bijection defined between them.
  We utilize this metric decomposition to prove a stability result for $\D$
  (we write $\D(X,Y)$ in short for $\D(\dgm(X),\dgm(Y))$ and similarly for $\dB$):
  {
    \renewcommand{\thetheorem}{\ref{thm:Hom:dN_stabil}}
    \begin{theorem} 
      $ \D(X,Y) \leq \frac{2\norm{\Delta}}{\diam(Y)}\,$,
    \end{theorem}
  }
\noindent where $\Delta$ represents the optimal metric decomposition of $Y$ in terms of $X$ for finite metric spaces $X,Y$ of same cardinality with a bijection between them (see Definition \ref{defn:Hom:met_scaling}) and $\diam(Y)$ is the diameter of $Y$.
We derive this stability result by explicitly bounding the distortion of the associated bijective map, and hence the Gromov-Hausdorff distance between $X$ and $Y$ (see Theorem \ref{thm:Hom:BnStb} and Definition \ref{defn:Hom:GH}).

We then study two popular dimension reduction techniques, Johnson-Lindenstrauss (JL) projections and metric multidimensional scaling (mMDS), and a third class of general biLipschitz mappings.
We provide new bounds on how well these dimension reduction techniques preserve homology with respect to $\D$.
For a JL map $f$ taking input $X$ to $f(X)$, we derive the following new bounds on $\dB$ and $\D$:
{
  \renewcommand{\thetheorem}{\ref{thm:Hom:JL_dB}}
  \begin{theorem}
    $\dB(X,f(X)) <\epsilon\op{diam}(X) \,$
  \end{theorem}
}
\noindent and
{
  \renewcommand{\thecorollary}{\ref{cor:Hom:JL_dN}}
  \begin{corollary}
    $\D(X,f(X))< \epsilon$\,,
  \end{corollary}
}
\noindent where $0 < \epsilon < 1$ is the tolerance up to which pairwise distances are preserved by $f$ (see Equation (\ref{eq:JLdef}) for the formal definition).
For mMDS, we present new bounds for both $\dB$ and $\D$ between persistence diagrams of $X$ and its projection in terms of the eigenvalues of the covariance matrix:
{
  \renewcommand{\thecorollary}{\ref{cor:Hom:mMDS_dB}}
  \begin{corollary}
    $\displaystyle \dB(X,P_X^{(m)}(X)) \leq \sqrt2 \sqrt[4]{\sum\limits_{m+1}^d \lambda_i^2}$,
  \end{corollary}
}
\noindent and
{
  \renewcommand{\thecorollary}{\ref{cor:Hom:mMDS_dN}}
  \begin{corollary}
    $\D(X,P^{(m)}_X(X)) \leq \frac{2\sqrt{2}}{\diam(P^{(m)}_X(X))}\sqrt[4]{\frac{\left(\sum\limits_{i=1}^m \lambda_i^2\right)\left(\sum\limits_{i=m+1}^d \lambda_i^2\right)}{\left(\sum\limits_{i=1}^d \lambda_i^2\right)}}$\,,
  \end{corollary}
}
\noindent where
$P_X^{(m)}:X \rightarrow \R^m$ is the mMDS projection (see Definition \ref{defn:Hom:mMDS}) and $\lambda_i$'s are the eigenvalues of the covariance matrix.
These bounds indicate that if the first $m$ eigenvectors capture most of the variance, and hence the remaining eigenvalues $\{\lambda_i\}_{m+1}^d$ are small, then homology is preserved well in terms of small values of both $\dB$ and $\D$.
And finally, we study a general class of $k$-biLipschitz projection maps $f:X\rightarrow Y$, i.e., which satisfy  $(1/k) \, d_X \, \leq d_Y \leq \, k \, d_X$ for some $k \geq 1$, to obtain bounds on homology preservation in terms of $\D$ as
{
  \renewcommand{\thecorollary}{\ref{cor:Hom:Lip_dN}}
  \begin{corollary}
    $\displaystyle \D{(X,Y)} \leq \left|\frac{k^2-1}{k}\right|\frac{\diam(X)}{\diam(Y)}$\,.
  \end{corollary}
}

\subsection{Related Work}

The use of bottleneck distance between persistence diagrams as a means of measuring similarity between metric spaces has gained much attention in the recent years.
In 2014, Chazal et al.~\cite{ChdeSOu2014} showed that the bottleneck distance between PDs of PCD is indeed stable with respect to the input metric spaces:
\begin{theorem}[Chazal et al., 2014~\cite{ChdeSOu2014}]
  \label{thm:Hom:BnStb}
  $\dB(X,Y)\leq 2 \dGH(X,Y)$,
\end{theorem}
where $\dGH$ is the Gromov-Hausdorff distance (see Definition \ref{defn:Hom:GH}).
Some researchers have recently considered modifications of the standard bottleneck distance between PDs that are invariant to certain transformations of the input metric spaces.
Sheehy et al.~\cite{CaKiSh2018} studied a \emph{shift-invariant bottleneck distance} between two PDs $A, B$ defined as
\[
  \op{d}_{\mathrm{sB}}(A, B)=\min _{s \in \mathbb{R}} \dB \left(A_{s}, B\right) \text{ where } A_s = \{s+a\ |\ a \in A\}.
\]
While shifts of PDs in the above manner will not directly capture scalings of the metric spaces, the authors mentioned that this shift-invariant similarity measure could possibly be applied to log-scale persistence diagrams in order to define a scale-invariant bottleneck distance.
At the same time, the question of scale invariance was not directly considered.
Cao et al.~\cite{CaVlScKaMo2021} were the first to study a dilation (i.e., scale) invariant bottleneck distance, where they defined the dissimilarity measure
\begin{equation} \label{defn:Hom:DilBn}
  \overline{\mathrm{d}_{\rm D}}\left(X, Y\right)= \inf_{c \in \mathbb{R}_{+}}
  \dB \left(cX, Y\right) \,.
\end{equation}
They show that this similarity measure is in fact dilation invariant and stable with respect to the Gromov-Hausdorff distance.
They further provide computational results for efficient searching for an optimal scaling factor $c$, as well as provide some computational evidence for the efficacy of this similarity measure.
On the other hand, we prove that an optimal decomposition of $Y$ in terms of $X$ always exists in our case (see Theorem \ref{thm:Hom:opt_decomp_exist}).

The question of preservation of topology of PCD under Johnson-Lindenstrauss (JL) random projections was addressed by Sheehy \cite{Sh2014} and Lotz \cite{Lo2019}.
They both showed that \v{C}ech persistent homology of a PCD is approximately preserved under JL random projections up to a $(1 \pm \epsilon)$ multiplicative factor for any $\epsilon \in [0,1]$ \cite{Sh2014,Lo2019}.
Sheehy proved this result for a PCD with $n$ points, whereas the result of Lotz considered more general sets of bounded Gaussian width, and also implied dimensionality reductions for sets of bounded doubling dimension in terms of the spread (ratio of the maximum to minimum interpoint distance).
Our bound of $\epsilon$ on $\D$ in Corollary \ref{cor:Hom:JL_dN} agrees with these results, while our proof approaches are distinct since we obtain a new bound on the distortion (see Definition \ref{defn:Hom:distor}) of the JL map in order to derive our bound.
Recently, the results of Sheehy and Lotz were extended by Arya et al.~\cite{ArBoDuLo2021} to the case of more general $k$-distances.
Similar bounds on homology preservation under other nonlinear DR techniques such as multidimensional scaling have been scarce. 

A motivation for our study of homology preservation under $k$-biLipschitz maps was the work of Mahabadi et al.~\cite{MaMaMaRa2018} who introduced the notion of outer biLipschitz extensions of maps between Euclidean spaces as a means for nonlinear dimension reduction.
But homology preservation was not a focus of their work.

The related problem of combining TDA with DR to design new topology-preserving DR techniques has been studied as well.
Motivated by the nonlinear DR technique of Isomap \cite{TedeSLa2000}, Yan et al.~\cite{YaZhRoScWa2018} proposed a method for homology-preserving DR via manifold landmarking and tearing.
Wagner et al.~\cite{WaSoBe2021} have used a topological loss term in a gradient-descent based method for DR that outperformed standard DR techniques in topology preservation.
But our focus is on quantifying topology preservation of existing DR techniques.

\subsection{Notation and Definitions}

We focus on comparing finite metric spaces denoted as $(X,d_X)$ and $(Y,d_y)$ unless otherwise stated, and neither $X$ nor $Y$ are trivial ($d_X$ and $d_Y$ are not identically $0$).
The size of a metric space is measured by its diameter, $\op{diam}(X) = \max\limits_{a,b \in X}d_X(a,b)$.

A common metric for comparing metric spaces is the Gromov-Hausdorff distance ($\dGH$), which can be thought of as a way of measuring to what degree two spaces are isometric.
We first define \emph{distortion}, and  define $\dGH$ in terms of distortion.



\begin{definition} \label{defn:Hom:distor}
The distortion of a relation $\sigma \in \mathcal{R}(X,Y)$ between metric spaces $X$ and $Y$ is
  \[
    \operatorname{dis} \sigma=\sup \left\{|d( x,x^{\prime})-d(y,y^{\prime})|:(x, y),\left(x^{\prime}, y^{\prime}\right) \in \sigma\right\}\,.
  \]
  The distortion of a mapping $f: X \rightarrow Y$, dis($f$), is the distortion of the graph of $f$.
\end{definition}


\begin{definition}
\label{defn:Hom:GH}
The Gromov-Hausdorff distance between metric spaces $X$ and $Y$ is
\[
  \dGH(X, Y)=\frac{1}{2} \inf \{\operatorname{dis} R: R \in \mathcal{R}(X, Y)\}\,.
\]
\end{definition}

We denote by $\dgm(X)$ the Vietoris-Rips persistence diagram of metric space $X$.
For a finite metric space $X$, this diagram consists of a finite multiset of points above the diagonal, i.e., $y=x$ line, in the extended plane $\bar{\R}^2$.
To this finite multiset, we add the infinitely many points on the diagonal, each with infinite multiplicity.
We now list the standard definition of the bottleneck distance between two such PDs \cite[\S VIII]{EdHa2009}.
\begin{definition}
\label{defn:dB}
The bottleneck distance between the persistence diagrams of two finite metric spaces $X$ and $Y$ is
\[
\dB(\dgm(X), \dgm(Y)) = \min_{\gamma: \dgm{X} \to \dgm{Y}} \max_{x \in \dgm{X}} \norm{x - \gamma(x)}_{\infty}~\text{for a bijection } \gamma.
\]
We shorten the notation to $\dB(X,Y)$ when there is no ambiguity.
\end{definition}

Finally, we collect in Table \ref{tab:not} some of the notation used throughout this paper.

\begin{table}[ht!]
\caption{Notation used throughout the paper, and their description.} \label{tab:not}
\begin{tabular}{|l|l|}
  \hline
  $\textnormal{VR}(X,r)$ & The Vietoris-Rips complex of metric space $X$ at {\bfseries diameter} $r$\\
  \hline
  $\dgm(X)$ & The Vietoris-Rips persistence diagram of metric space $X$\\
  \hline
  $\dB(X,Y)$ & Bottleneck distance between $\dgm(X)$ and $\dgm(Y)$\\
  \hline
  $\diam_X(A)$ & The diameter of a set $A\subset X$ with respect to a metric $d_X$, $\max\limits_{a,b \in A}\{d_X(a,b)\}$\\
  \hline
  $s\cdot A$ & $\{s\cdot a\ |\ a \in A\}$ when $A \subset \R^n$\\
  \hline
  $C(Z)$ & The covariance matrix $ZZ^T$ of $Z$, when $Z = [z_1 \dots\ z_n]$ is a data matrix\\
  \hline
  $A \circ B$ & The Hadamard product of $A$ and $B$\\
  \hline
  $\displaystyle A^{\circ^n}$ & The $n^{\text{th}}$ Hadamard power of a matrix $A$\\
  \hline
  $G_X$ & The gram matrix of a data set $X \subset \R^n$\\
  \hline
  $I_{|_m}$ & The matrix: $\op{diag}(\underbrace{1,1,\ldots,1}_{m},\underbrace{0,0,\ldots,0}_{n-m})$ for $m\leq n$\\
  \hline
  $\Sigma_{|_m}$ & The matrix $\Sigma\cdot I_{|_m}$\\
  \hline
\end{tabular}  
\end{table}

\section{Methods and Results}

Given a metric space $(X,d_X)$ we consider the family of metric spaces $sX := (X,sd_X)$ for  $s>0$.
In words, we consider a family of metrics on $X$ that are positive scalars of $d_X$.
Intuitively, the family $sX$ can be thought of as all dilations and contractions of $X$.

\begin{proposition}
\label{prop:Hom:MSprpts}
The following properties hold for a metric space $(X,d_X)$ and $p,q > 0$.
 \begin{enumerate}
   \item $\diam(pX) = p\diam(X)$.
   \item $p(qX) = (pq)X$.
 \end{enumerate}
\end{proposition}

Scaling metric spaces has natural consequences on VR complexes, persistence diagrams, and bottleneck distances.

\begin{lemma}
\label{lem:Hom:VR_scaling}
Let $s>0$.
Then $\op{VR}(X,r) = \op{VR}(sX,sr)$ for all $r>0$.
\end{lemma}
\begin{proof}
Let $s,r >0$.
By definition, $\op{VR}(X,r) = \{\sigma \subset X\ |\ \op{diam}_X(\sigma) < r\}$, so let $\sigma \in \VR{X}{r}$.
Then $\op{diam}_X(\sigma)<r \iff s\op{diam}_X(\sigma)<sr \iff \op{diam}_X(s\sigma)<sr \iff  \op{diam}_{sX}(\sigma)<sr \iff \sigma \in \VR{sX}{sr}$.
\end{proof}
\begin{corollary}
\label{cor:Hom:dgm_scaling}
$s\cdot \dgm(X) = \dgm(sX)$ for all $s>0$.
\end{corollary}
\begin{proof}
Follows immediately from Lemma \ref{lem:Hom:VR_scaling}.
\end{proof}

\begin{theorem}
  \label{thm:Hom:dB_scaling}
  $s\cdot \dB(X,Y) = \dB(sX,sY)$ for all $s>0$.
\end{theorem}
\begin{proof}
  \[
     \arraycolsep=1.4pt
     \def\arraystretch{2}
     \begin{aligned}
       s\cdot \dB(X,Y) & =  s \left[\min\limits_{\gamma:\dgm(X) \rightarrow \dgm(Y)} \, \max\limits_{x \in \dgm(X)} \, \norm{x-\gamma(x)}\right] \\
       & = \min\limits_{\gamma:\dgm(X) \rightarrow \dgm(Y)} \, \max\limits_{x \in \dgm(X)} \, s\norm{x-\gamma(x)} \\
       & = \min\limits_{\gamma:\dgm(X) \rightarrow \dgm(Y)} \, \max\limits_{x \in \dgm(X)} \, \norm{sx-s\gamma(x)} .
     \end{aligned}
  \]
Applying Corollary \ref{cor:Hom:dgm_scaling}, we have that $sx \in \dgm(sX)$ and $s\gamma(x) \in \dgm(sY)$, and hence the 

\smallskip
\noindent above expression is equivalent to $\displaystyle \min\limits_{\gamma:\dgm(sX) \rightarrow \dgm(sY)} \, \max\limits_{x \in \dgm(sX)}\norm{x-\gamma(x)} = \dB(sX,sY) $ by definition.
\end{proof}

We define the normalized bottleneck distance using the notion of metric scaling.

\begin{definition}[Normalized Bottleneck]
\label{defn:Hom:dN}
$\displaystyle \D(X,Y) = \dB\left(\frac{X}{\diam(X)},\frac{Y}{\diam(Y)}\right)$.
\end{definition}

It is easy to see that $\D$ is indeed a pseudometric on persistence diagrams.
It inherits most properties of $\dB$, except it is possible for $\D$ to be zero for two distinct PDs.
We can see that $\D$ fails the identity of indiscernibles by taking a metric space $X$ and scaled version $sX$ for $s > 0$.
Then $\D(sX,X) = \D(X,X) = 0$, despite $sX$ and $X$ being distinct for $s \neq 1$.

\begin{theorem}[Scale-invariance]
\label{thm:Hom:dN_scaleIn}
  $\displaystyle \D(pX,qY) =\D(X,Y)$ for all $p,q >0$.
\end{theorem}
\begin{proof}
Let $p,q>0$.
Then
\[
  \arraycolsep=1.4pt
  \def\arraystretch{2}
  \begin{aligned}
    \D(pX,qY) & = 
    \dB\left(\frac{pX}{\diam(pX)},\frac{qY}{\diam(qY)}\right)\\
    & = 
    \dB\left(\frac{p}{p}\frac{X}{\diam(X)},\frac{q}{q}\frac{Y}{\diam(Y)}\right) \\
    & = 
    \dB\left(\frac{X}{\diam(X)},\frac{Y}{\diam(Y)}\right) \\
    & =  \D(X,Y).
  \end{aligned}
\]
\end{proof}

\subsection{Metric Decomposition}

When $X$ and $Y$ are of equal size and we have a bijection $f : X \rightarrow Y$, we can examine the metrics of both $X$ and $Y$ via distance matrices $D_X$ and $D_Y$, where $[D_X]_{ij} := d_X(x_i,x_j)$ and $[D_Y]_{ij} := d_Y(f(x_i),f(x_j))$.
This is a typical setting encountered in dimension reduction where $Y$ is a lower dimensional projection of $X$.
Note that we must have such a map $f$ to compare $D_X$ and $D_Y$, since distance matrices are only unique up to orderings of the points.
We have that $\diam(X) = \maxnorm{D_X}$.
We omit the subscript and write $\norm{\cdot}$ to mean $\maxnorm{\cdot}$ when it is clear.

\begin{definition}
  \label{defn:Hom:met_scaling}
  Let $f: X \rightarrow Y$ be a bijection, $s\geq0$.
  Define $\Delta_{s}=\left[\delta_{i j}\right]_{s}=D_{Y}-s D _X$.
  We say that $s D_{X}+\Delta_{s}=D_{Y}$ is a metric decomposition of $Y$ in terms of $X$.
\end{definition}

There are infinitely many choices for $s$ and hence infinitely many decompositions, so naturally we seek an optimal decomposition.
Thinking about $\Delta_s$ as the ``error'' between $D_X$ and $D_Y$, we seek to minimize this error.
In other words, we seek an $s$ such that $\maxnorm{\Delta_s}$ is minimized.
It is not immediately clear that such a scalar $s$ exists, but indeed this is a well-posed problem.

\begin{lemma}
  \label{lem:Hom:error_norm_propts}
  Let $X,Y$ be metric spaces, $f:X\rightarrow Y$ bijective, and let $h:\R_+ \rightarrow \R$ be the map $s \mapsto \maxnorm{\Delta_s}$.
  The following properties hold:
  \begin{enumerate}
    \item $h$ is continuous; \label{lem:Hom:error_norm_propts_hcntns}
    \item there exists an $m >0$ such that $h$ is strictly increasing on $[m,\infty)$. \label{lem:Hom:error_norm_propts_hincr}
  \end{enumerate}
\end{lemma}
\begin{proof}
  \begin{enumerate}
    \item This follows from the continuity of $\maxnorm{\cdot}$.
      Let $\epsilon >0$, and set $\displaystyle \delta = \frac{\epsilon}{\maxnorm{D_X}}$.
      Then, if $|s-a|< \delta$, we have
      \[
         \begin{aligned}
           |h(s)-h(a)| & = \left|\,\maxnorm{\Delta_s}-\maxnorm{\Delta_a} \,\right| \\
                       & = \left|\,\norm{D_Y - sD_X}-\norm{D_Y-aD_X} \, \right|.
         \end{aligned}
      \]
      Applying the reverse triangle inequality, we obtain
      \[
         \begin{aligned}
           |h(s)-h(a)| & \leq \norm{(s-a)D_X}  \\
           & = |s-a|\norm{D_X} \\
           & < \delta\norm{D_X} \\
           & = \frac{\epsilon}{\maxnorm{D_X}}  \maxnorm{D_X} ~ = \epsilon \,.
         \end{aligned}
      \]
    \item Let $m >0$ be such that $m d_X(a,b) - d_Y(f(a),f(b)) >0$ for all $a,b \in X$, and let $u,v\in X$ achieve the maximum $m d_X(u,v) - d_Y(f(u),f(v)) = \max\limits_{a,b} \{m d_X(a,b) - d_Y(f(a),f(b))\} = \maxnorm{\Delta_m}$.
      If $p > q \geq m$, then clearly $u,v$ also achieve the maximums for $\maxnorm{\Delta_p}$ and $\maxnorm{\Delta_q}$, so necessarily
      \[
         \begin{aligned}
           h(p)-h(q) & = \maxnorm{\Delta_p}-\maxnorm{\Delta_q} \\
           & = p\cdot d_X(u,v) - d_Y(f(u),f(v)) -(q\cdot d_X(u,v) - d_Y(f(u),f(v))) \\
           & = (p-q)d_X(u,v) >0 ~\text{ since } p-q>0 \,.
         \end{aligned}
      \]
      Thus, $h(p)-h(q) >0$.
  \end{enumerate}
\end{proof}

\begin{theorem}
  \label{thm:Hom:opt_decomp_exist}
  Let $\displaystyle D_Y = sD_X + \Delta_s$ for any $s>0$.
  Then $\min\limits_s \maxnorm{\Delta_s}$ exists.
\end{theorem}
\begin{proof}
Let $h:\R_+ \rightarrow \R$ be the mapping from Lemma \ref{lem:Hom:error_norm_propts}.
Using Property \ref{lem:Hom:error_norm_propts_hincr} in Lemma \ref{lem:Hom:error_norm_propts}, let $m$ be a marker such that $h$ is strictly increasing on $[m,\infty)$.
Then clearly $\displaystyle \restr{h}{[m,\infty)}$ has a global minimum at $m$.
Likewise, $\displaystyle \restr{h}{[0,m]}$ has a global minimum since $h$ is continuous and $[0,m]$ is compact.
Now, since $\R_+ = [0,m] \cup [m,\infty)$ and $h$ achieves a minimum of both intervals, $h$ has global minimum $\displaystyle \min\{\min \restr{h}{[m,\infty)},\ \min \restr{h}{[0,m]}\}$.
\end{proof}

Hence we can always minimize $\maxnorm{\Delta_s}$.
In general, we let $s^* = \argmin\limits_s \maxnorm{\Delta_s}$ be an optimal decomposition of $Y$ in terms of $X$, and write $\Delta_{s^*} = \Delta$.

\subsection{Stability Result}

Since $\D$ and $\dB$ are similarity measures on metric spaces that measure the spaces indirectly via their respective persistence diagrams, one would desire that $\D$ and $\dB$ are ``stable'' with respect to the input metric spaces.
Stability in this sense means that
\begin{itemize}
    \item ``small changes'' to $X$ or $Y$ result in ``small changes'' to $\dB$ and $\D$, and
    \item if $X$ and $Y$ are ``close'', then $\D$ and $\dB$ are also ``close''.
\end{itemize}
To make these notions precise, small changes in $X$ or $Y$ and the closeness of $X$ and $Y$ are measured using some other metric for $X$ and $Y$, preferably one that compares the metric spaces directly.
A good choice for this metric is the Gromov-Hausdorff distance, $\dGH$, specified in Definition \ref{defn:Hom:GH},
Recall the bound on $\dB$ in terms of $\dGH$ given by Chazal et al.~\cite{ChdeSOu2014} in Theorem \ref{thm:Hom:BnStb}: 
$\dB(X,Y)\leq 2\dGH(X,Y)$.
When we have a bijective correspondence between $X$ and $Y$ and hence a metric decomposition, we obtain a much stronger stability bound.

\begin{theorem}
  \label{thm:Hom:dN_stabil}
  \[ \D(X,Y) \leq \frac{2\norm{\Delta}}{\diam(Y)} \,. \]
\end{theorem}
\begin{proof}
Let $\displaystyle f :X \rightarrow Y$ be the bijection and $\displaystyle \hat{f}: \frac{X}{\norm{D_X}} \rightarrow \frac{Y}{\norm{D_Y}}$.
Observe that \par \vspace{2mm}
\[
\D(X,Y) = \dB\left(\frac{X}{\norm{D_X}},\frac{Y}{\norm{D_Y}}\right) \leq \dis (\hat{f}) \,.
\]
Using our decomposition,
\[
  \begin{aligned}
    \op{dis}(\hat{f}) & = \max\limits_{a,b \in X}\left| \frac{d_X(a,b)}{\norm{D_X}} - \frac{sd_X(a,b) +\delta_s(a,b)}{\norm{D_Y}} \right|\\
    & = \max\limits_{a,b \in X}\frac{\big| d_X(a,b)\norm{sD_X+\Delta_s}-\norm{D_X}(sd_X(a,b) +\delta_s(a,b))\big|}{\norm{D_X}\norm{D_Y}}\\
    & = \max\limits_{a,b \in X}\frac{\big|(\norm{sD_X + \Delta_s}-\norm{sD_X})d_X(a,b)-\norm{D_X}\delta_s(a,b)\big|}{\norm{D_X}\norm{D_Y}}\\
    & \leq \max\limits_{a,b \in X}\frac{\big|\norm{sD_X+\Delta_s}-\norm{sD_X}\big|d_X(a,b) + \norm{D_X}|\delta_s(a,b)|}{\norm{D_X}\norm{D_Y}} ~~\text{ by triangle inequality.}
  \end{aligned}
\]
Observe that $\big| \norm{sD_X+\Delta_s}-\norm{sD_X} \big| \leq \norm{\Delta_s}$ via the reverse triangle inequality, so,
\[
  \begin{aligned}
    \D(X,Y) & \leq \max\limits_{a,b \in X}\frac{\norm{\Delta_s}d_X(a,b) + \norm{D_X}|\delta_s(a,b)|}{\norm{D_X}\norm{D_Y}} \\
    & \leq \max\limits_{a,b \in X}\frac{2\norm{\Delta_s}\norm{D_X}}{\norm{D_X}\norm{D_Y}} \\
    & = 2 \frac{\norm{\Delta_s}}{\norm{D_Y}} \, .
  \end{aligned}
  \]
  This holds for all $s$, and thus holds when $\Delta_s = \Delta$.
\end{proof}

\section{Applications in Dimension Reduction}

Dimension reduction (DR) techniques are used ubiquitously and are paramount for discovering latent features of data and visualizing high-dimensional data.
However, the curse of dimensionality dictates that distances tend to grow as the number of dimension grow, which can lead to a proportionate amount of scaling under certain DR techniques.
But more specifically in our context, DR techniques are necessarily bijective between its respective metric spaces, and thus exhibit metric decompositions.
Hence we seek to derive guarantees for homology preservation under DR in terms of bounds on $\D$ between the PDs of the input and reduced metric spaces.
We look at two specific dimension reduction techniques, Johnson-Lindenstrauss projections and metric multidimensional scaling, and provide specific bounds for $d_B$ and $d_N$.
We also provide a more general condition for biLipschitz functions that may serve other dimension reduction techniques.

\subsection{Johnson-Lindenstrauss Linear Projection}

A Johnson-Lindenstrauss (JL) linear projection \cite{JoLi1984} is, for some $0 < \epsilon < 1$, a linear map $f:X \rightarrow Y$ satisfying the inequality
\begin{equation} \label{eq:JLdef}
  (1-\epsilon)d_X(u,v)^{2} \leq d_Y(f(u),f(v))^{2} \leq(1+\epsilon)d_X(u,v)^{2} \text{ for all } u,v \in X\,.
\end{equation}
A mapping that satisfies this property may be desirable in many cases, as it guarantees pairwise distances do not change ``too much'' under $f$.
Such a mapping is guaranteed to exist from $\R^N$ to $\R^n$ for $n \leq N$ provided $n$ is not too small.
This result was formalized in the classical Johnson-Lindenstrauss Lemma.
\begin{lemma}[Johnson and Lindenstrauss, 1984~\cite{JoLi1984}]
\label{lem:Hom:JL}
Given $0<\epsilon<1$, a set $X$ of $m$ points in $\mathbb{R}^{N}$, and a number $n>8 \ln (m) / \epsilon^{2}$, there is a linear map $f: \mathbb{R}^{N} \rightarrow \mathbb{R}^{n}$ such that
\[
  (1-\epsilon)\|u-v\|^{2} \leq\|f(u)-f(v)\|^{2} \leq(1+\epsilon)\|u-v\|^{2} ~\text{ for all } u, v \in X\,.
\]
\end{lemma}

Because pairwise distances under a $\JL$ mapping do not change too much, we might expect a small difference in persistence as well, i.e., a small bottleneck distance.
However, this pairwise distance bound is multiplicative with respect to $d_X$, which leads to a potentially large bottleneck distance.

\begin{lemma}
\label{lem:Hom:JL_lem}
 Let $ \epsilon >0$.
 If $\displaystyle \left|\frac{\op{d}(f(x),f(x'))}{\op{d}(x,x')}-1 \right| < \epsilon$ for all $x,x' \in X$, then
 \[ \op{dis}(f)< \epsilon\op{diam}(X) \,. \]
\end{lemma}
\begin{proof}
Let $x,x' \in X$, where $(1-\epsilon)\op{d}(x,x')< d(f(x),f(x')) < (1+\epsilon)\op{d}(x,x')$.
Subtracting $\op{d}(x,x')$ gives
\[ -\epsilon\op{d}(x,x')< d(f(x),f(x')) - \op{d}(x,x')< \epsilon\op{d}(x,x') \, .\]
Hence $\left| \op{d}(f(x),f(x')) - \op{d}(x,x') \right| < \epsilon \op{d}(x,x')$.
Taking the supremum over all $x,x' \in X$ on both sides gives us our result.
\end{proof}

\begin{theorem}[Johnson-Lindenstrauss Homology Preservation]
  \label{thm:Hom:JL_dB}
  Let $X \subset \R^N$ with $|X|=m$, and let $\epsilon \in (0,1)$.
  If $f:\R^N \rightarrow \R^n$ is a JL-linear map ($n>8 \ln (m) / \epsilon^{2}$), then
  \[ \dB(X,f(X)) <\epsilon\op{diam}(X) \,. \]
\end{theorem}
\begin{proof}
  Let $0<\epsilon<1$, and $f$ be a Johnson-Lindenstrauss projection with $n > 8 \ln (m) / \epsilon^{2}$.
  Then by the Johnson-Lindenstrauss lemma (Lemma \ref{lem:Hom:JL}), we have
  \[
  (1-\epsilon)\|x-x'\|^{2} \leq\|f(x)-f(x')\|^{2} \leq(1+\epsilon)\|x-x'\|^{2} \,.
  \]
  Taking positive roots, we have
  \[
  \sqrt{1-\epsilon}\|x-x'\| \leq\|f(x)-f(x')\| \leq\sqrt{1+\epsilon}\|x-x'\|\,,
  \]
  and since $1-\epsilon < \sqrt{1-\epsilon}$ and $\sqrt{1+\epsilon}<1+\epsilon$, we arrive at
  \[
  (1-\epsilon)\|x-x'\| \leq\|f(x)-f(x')\| \leq(1+\epsilon)\|x-x'\|\,.
  \]
  Hence, $\displaystyle \left| \frac{\|f(x)-f(x')\|}{\|x-x'\|}-1\right|<\epsilon$.
  Applying Lemma \ref{lem:Hom:JL_lem} gives us the desired result.
\end{proof}

We can see that the bound in Theorem \ref{thm:Hom:JL_dB} can be large when $\diam(X)$ is large.
But this is not an issue with the normalized bottleneck since our spaces are normalized.

\begin{corollary}
  \label{cor:Hom:JL_dN}
  For a JL-Linear map, $\D(X,f(X))< \epsilon$\,.
\end{corollary}

\subsection{Metric Multidimensional Scaling}

Metric multidimensional scaling (mMDS) is a generalization of the classical multidimensional scaling (MDS) \cite{CoCo2000}.
We introduce the framework of mMDS, prove a new result on the projection map of mMDS (Lemma \ref{lem:Hom:mMDS_lem}), and use it to derive new bounds on $\dB$ and $\D$ between the input and output spaces of mMDS.
The mMDS framework orthogonally projects data onto a subspace chosen to minimize mean-squared error or to maximize variance.
More precisely, given an input metric space $X=(\{x_1,\dots,x_n\},d_X)$ and desired reduced dimension $m$, we find a centered data set $\Tilde{X} = \{\Tilde{x}_1,\ldots,\Tilde{x}_n\}\subset \R^m$ such that $\sum\limits_{i,j=1}^n(d(x_i,x_j)-\norm{\Tilde{x_i}-\Tilde{x}_j})^2$ is minimized.
Such a projection can be achieved in two steps.
\begin{enumerate}
    \item \label{mMDSRlzn} Obtain a realization $\Phi$ of $X$ in $\R^k$ for some $k$, i.e., $\Phi :X \rightarrow \Phi(X) \subset \R^k$ is an isometry.
    \item Orthogonally project the realized data onto the first $m \leq d$ dominant eigenvectors of the covariance matrix $C(\Phi(X))$.
\end{enumerate}

To find a realization of $X$ in $\R^k$, recall that $-\frac{1}{2}D_X^{\circ^2} = C_nG_XC_n$, where $G_X$ is the Gram matrix of $X$, and $C_n = I_n-\frac{1}{n}\mathbf{1}_n$ is the centering matrix.
Since $G_X$ is positive semidefinite, it has a unique root $Z = [z_1 \dots z_n] \in \R^{k \times n}$ so that $Z^TZ = G_X$, for some $k$.
The realization of $X$ are the columns of $Z$, i.e., $\Phi(X) = \{z_1, \dots, z_n\}$.

With our realization matrix $Z$ we can compute the singular-value decomposition  $Z = U\Sigma V^T$, where we choose $\Sigma = \op{diag}(\sigma_1,\dots,\sigma_n)$ with $\sigma_1 \geq \ldots \geq \sigma_n$.
The eigenvectors of $C(Z) = ZZ^T = U\Sigma^2U^T$ are precisely the columns of $U$, and the eigenvalues $\lambda_i$ of $C(Z)$ are precisely $\sigma_i^2$.
Taking the first $m$ columns of $U$, we perform an orthogonal projection.

\begin{definition}[metric multidimensional scaling (mMDS)]
  \label{defn:Hom:mMDS}
  Let $X=(\{x_i\},d_X),\ i=1, \dots, n$ be a metric space, $G_X = U\Sigma^2 U^T$ the corresponding Gram matrix, and $Z = [z_1 \cdots z_n]$ a realization of $X$ in $\R^k$ where $G_X = Z^TZ$.
  Let $0 < m \leq d$, and $\Tilde{U} = \left[u_1 \cdots  u_m\right]$ be the truncated matrix consisting of the first $m$ dominant eigenvectors of $G_X$.
  The mMDS reduction $P_X^{(m)} : X \rightarrow \R^m$ is the map $P_X^{(m)}(x_i) = \Tilde{x_i}= \op{proj}_{\op{Im}(\Tilde{U})}z_i$.
\end{definition}

We assume that the realization $\Phi$ in Step \ref{mMDSRlzn} above is attainable, even if this computation may be hard (in fact, such a realization may not be guaranteed to exist in all cases).
More generally, the problem of mMDS is NP-hard but good approximation algorithms have been proposed \cite{DeHeKoLyUr2021}, and it is considered to work well in practice \cite{CoCo2000}.
Our focus is on deriving bounds on homology preservation by mMDS with respect to $\dB$ and $\D$.

\begin{remark}
  $P_X^{(m)}$ can equivalently be represented in matrix form as $P \in \R^{k \times k}$ so that $\Tilde{x_i} = P z_i$.
  The corresponding data matrix of the image $\displaystyle P_X^{(m)}(\{x_i\})$ is then $\Tilde{X} := [P z_1  \cdots P z_n] = PZ$.
  In fact, it can be shown that $\displaystyle P = \Tilde{U}(\Tilde{U}^T\Tilde{U})^{-1}\Tilde{U}^T$.
\end{remark}

\begin{lemma}
  \label{lem:Hom:mMDS_lem}
  Let $Z = U\Sigma V^T$ be the singular value decomposition (SVD) of $Z$.
  Then $P$ can be diagonalized as $P = UI_{|_m}U^T$, and $\Tilde{X} = U \Sigma_{|_m} V^T$
  {\rm (see Table \ref{tab:not} for $I_{|_m}, \Sigma_{|_m}$ notation)}.
\end{lemma}

\begin{proof}
  Since $P$ orthogonally projects onto the image of the first $m$ dominant eigenvectors of $C(X)$, We necessarily have $Pu_i = u_i$ for $0 \leq i \leq m$ and $Pu_i = 0$ for $m<i\leq d$.
Thus $P$ and $C(X)$ share eigenvectors, and the spectrum of $P$ consists of $m$ ones and $d-m$ zeros.
Thus, P is diagonalizable, and if we order the eigenvalues/eigenvectors in decreasing order, we have $P = UI_{|_m}U^T$.
Now, if $X = U\Sigma V^T$, then $\Tilde{X} = PX = (UI_{|_m}U^T)(U\Sigma V^T) = (UI_{|_m})(U^TU)(\Sigma V^T) = U(I_{|_m}\Sigma )V^T = U\Sigma_{|_m} V^T$.
\end{proof}

Unlike JL mappings, mMDS minimizes a more global measure of pairwise distance
\[
  J := \sum\limits_{i,j=1}^n(d(x_i,x_j)-\norm{\Tilde{x_i}-\Tilde{x_j}})^2 \,.
\]
The minimum $J$ is related directly to the eigenvalues of the covariance matrix $C(X)$:
\[\min\limits_{\Tilde{X}\subset \R^m}J = \sum\limits_{i=m+1}^{d}\lambda_i^2\ .\]
We now show that a similar relationship holds for the distortion of $P^{(m)}_X$ that depends only on the eigenvalues of the covariance matrix.

\begin{theorem}
  \label{thm:Hom:mMDS_dist}
  $\displaystyle \op{dis}(P_X^{(m)}) \leq \sqrt2 \sqrt[4]{\sum\limits_{m+1}^d \lambda_i^2}$\,.
\end{theorem}

We first prove a result on Hadamard product $A \circ A$ of square matrices $A$ with themselves.

\begin{lemma}
\label{lem:Hom:Hada_lem}
Let $A = [a_{ij}],B = [b_{ij}] \in \R^{n\times n}$ be square matrices with non-negative entries.
Then $\maxnorm{A-B}^2 \leq \maxnorm{A^{\circ^2}-B^{\circ^2}}$.
\end{lemma}

\begin{proof}
\[
  \begin{aligned}
    \maxnorm{A-B}^2 & = (\max |a_{ij}-b_{ij}|)^2 \\
    &  =    \max |a_{ij}-b_{ij}|^2 \\
    & \leq  \max |a_{ij}-b_{ij}||a_{ij}+b_{ij}| \\
    &  =    \max |(a_{ij}-b_{ij})(a_{ij}+b_{ij})| \\
    &  =    \max |a_{ij}^2-b_{ij}^2| \\
    &  =    \maxnorm{A^{\circ^2}-B^{\circ^2}} \, .
  \end{aligned}
\]
\end{proof}

\begin{proof}[Proof of Theorem \ref{thm:Hom:mMDS_dist}]
  We note that $\op{dis}(P_X^{(m)})^2 = \maxnorm{D_X-D_{\Tilde{X}}}^2 = \maxnorm{D_Z-D_{\Tilde{X}}}^2$.
  Applying Lemma \ref{lem:Hom:Hada_lem}, we get $\op{dis}(P_X^{(m)})^2 \leq \maxnorm{D_Z^{\circ^2}-D_{\Tilde{X}}^{\circ^2}}$.
  Now, recall that Euclidean distance matrices and corresponding Gram matrices are related as $Z^TZ =  -\frac{1}{2}C_kD_Z^{\circ^2}C_k$.
  Since $C_k$ is idempotent we can rearrange this so that $-\frac{1}{2}D_Z^{\circ^2} = C_kZ^TZC_k = C_kV\Sigma^2V^TC_k = (C_kV)\Sigma^2(V^TC_k^T)$.
  Now since $C_k$ and $V$ are both unitary, $Q := C_kV$ is also unitary, so we can write $-\frac{1}{2} D_Z^{\circ^2} = Q\Sigma^2Q^T$.
  Using the same argument for $D_{\Tilde{X}}^{\circ^2}$ and applying Lemma \ref{lem:Hom:mMDS_lem}, we get $-\frac{1}{2} D_{\Tilde{X}}^{\circ^2} = Q\Sigma_{|_m}^2Q^T$.
  So now,
  \[
     \begin{aligned}
       \maxnorm{D_X^{\circ^2}-D_{\Tilde{X}}^{\circ^2}} & =  2\maxnorm{Q\Sigma^2Q^T-Q\Sigma_{|_m}^2Q^T} \\
       & \leq 2\norm{Q\Sigma^2Q^T-Q\Sigma_{|_m}^2Q^T}_F ~~\text{(Frobenius norm)}\\
       &  =   2\norm{Q(\Sigma^2-\Sigma_{|_m}^2)Q^T}_F \, .
     \end{aligned}
  \]

  Since $\norm{\cdot}_F$ is unitary invariant, the above expression is
  \[
    \begin{aligned}
      \hspace*{1in} & = 2\norm{\Sigma^2-\Sigma_{|_m}^2}_F = 2\sqrt{\sum\limits_{i=m+1}^d \lambda_i^2} \, .
    \end{aligned}
  \]
  Bringing everything together, we arrive at
  \[
    \op{dis}(P_X^{(m)})^2 ~\leq~ 2\sqrt{\sum\limits_{i=m+1}^d \lambda_i^2} \,.
  \]
  Taking the square root on both sides gives us the desired result.
\end{proof}

\begin{corollary}
  \label{cor:Hom:mMDS_dB}
  $\displaystyle \dB(X,P_X^{(m)}(X)) \leq \sqrt2 \sqrt[4]{\sum\limits_{m+1}^d \lambda_i^2}\,$.
\end{corollary}
\begin{proof}
  Follows directly from Theorem \ref{thm:Hom:mMDS_dist}, Theorem \ref{thm:Hom:BnStb}, and Definition \ref{defn:Hom:GH}.
\end{proof}

We now derive a corresponding bound for $\D$ under mMDS, which could be much smaller than the bound on $\dB$ in Corollary \ref{cor:Hom:mMDS_dB}.
We first bound $\maxnorm{\Delta}\,$.
\begin{theorem}
  \label{thm:Hom:mMDS_opt_decomp}
  $\displaystyle \maxnorm{\Delta} \leq \sqrt{2}\sqrt[4]{\frac{\left(\sum\limits_{i=1}^m \lambda_i^2\right)\left(\sum\limits_{i=m+1}^d \lambda_i^2\right)}{\sum\limits_{i=1}^d \lambda_i^2}}\,$.
\end{theorem}
\begin{proof}
  $\maxnorm{\Delta_s}^4 = \maxnorm{D_Z-sD_{\Tilde{X}}}^4$.
  Applying Lemma \ref{lem:Hom:Hada_lem}, we observe that
  \[
     \begin{aligned}
       \maxnorm{\Delta_s}^4 & \leq \maxnorm{D_Z^{\circ^2}-(sD_{\Tilde{X}})^{\circ^2}}^2 \\
       & = \maxnorm{D_Z^{\circ^2}-s^2D_{\Tilde{X}}^{\circ^2}}^2 \\
       & = 4\maxnorm{Q\Sigma_{|_m}Q^T - s^2Q\Sigma Q^T}^2 \\
       & \leq 4\norm{\Sigma_{|_m} - s^2\Sigma}_F^2 \\
       & = 4\left(\sum\limits_{1}^{m}\lambda_i^2(1-s^2)^2 + \sum\limits_{m+1}^{d}s^4\lambda_i^2\right)\,.
     \end{aligned}
  \]

Now, since $\maxnorm{\Delta}\leq \maxnorm{\Delta_s}$ for all $s>0$, we have:
\[
  \frac{1}{4}\maxnorm{\Delta}^4\leq (1-s^2)^2\sum\limits_{1}^{m}\lambda_i^2 + s^4\sum\limits_{m+1}^{d}\lambda_i^2 ~\text{ for all } s.
\]
The right-hand side is a polynomial in $s$ and is minimized when $s = \displaystyle \sqrt{\frac{\sum\limits_1^m \lambda_i^2}{\sum\limits_1^d \lambda_i^2}}$.
For convenience, let $\sig{p}{q} := \sum\limits_{p}^{q}\lambda_i^2$, so this expression is minimized at $\displaystyle s = \sqrt{\frac{\sig{1}{m}}{\sig{1}{d}}}$.
Plugging this in, we obtain

\[
   \begin{aligned}
     \frac{1}{4}\maxnorm{\Delta}^4 & \leq \left(1-\frac{\sig{1}{m}}{\sig{1}{d}}\right)^2 \sum\limits_{1}^{m}\lambda_i^2 + \left(\frac{\sig{1}{m}}{\sig{1}{d}}\right)^2\sum\limits_{m+1}^{d}\lambda_i^2 \\
     & = \left(1-\frac{\sig{1}{m}}{\sig{1}{d}} \right)^2\sig{1}{m} + \left(\frac{\sig{1}{m}}{\sig{1}{d}} \right)^2\sig{m+1}{d} \\
     & = \left(\frac{\sig{m+1}{d}}{\sig{1}{d}}\right)^2\sig{1}{m} + \left(\frac{\sig{1}{m}}{\sig{1}{d}} \right)^2\sig{m+1}{d} \\
     & = \left(\frac{\sig{1}{m}\sig{m+1}{d}}{(\sig{1}{d})^2} \right)\sig{1}{d} \\
     & = \frac{\sig{1}{m}\sig{m+1}{d}}{\sig{1}{d}} \, .
   \end{aligned}
\]

Thus,
\[
   \maxnorm{\Delta}^4 ~\leq~ 4\frac{\left(\sum\limits_{i=1}^m \lambda_i^2\right)\left(\sum\limits_{i=m+1}^d \lambda_i^2\right)}{\sum\limits_{i=1}^d \lambda_i^2} \,.
\]
Quartic-rooting both sides gives the desired result.
\end{proof}

\begin{corollary}
\label{cor:Hom:mMDS_dN}
  $\D(X,P^{(m)}_X(X)) \leq \frac{2\sqrt{2}}{\diam(P^{(m)}_X(X))}\sqrt[4]{\frac{\left(\sum\limits_{i=1}^m \lambda_i^2\right)\left(\sum\limits_{i=m+1}^d \lambda_i^2\right)}{\left(\sum\limits_{i=1}^d \lambda_i^2\right)}}\,$.
\end{corollary}

\begin{proof}
This follows from applying Theorem \ref{thm:Hom:mMDS_opt_decomp} to Theorem \ref{thm:Hom:dN_stabil}.
\end{proof}

\subsection{BiLipschitz Functions}

Finally, we present bounds for $\D$ under a general class of biLipschitz mappings.

\begin{theorem}
  \label{thm:Hom:Lip_opt_decomp}
  Let $f:X \rightarrow Y$ be $k$-biLipschitz.
  Then $\displaystyle \norm{\Delta} \leq \left|\frac{k^2-1}{2k}\right|\norm{D_X}\,$.
\end{theorem}
\begin{proof}
  Since $f$ is $k$-biLipschitz we have that $\displaystyle \frac{1}{k}d_X \leq d_Y \leq kd_X$ for some $k \geq 1$.
  Using our metric decomposition, this becomes $\displaystyle \frac{1}{k}d_X \leq sd_X + \lambda_s \leq kd_X$.
  Rearranging we have $\displaystyle \frac{1}{k}-s \leq \frac{\lambda_s}{d_X} \leq k-s$.
  Choosing $\displaystyle s = \frac{\frac{1}{k}+k}{2}$ further gives us $\displaystyle \frac{1}{2k}-\frac{k}{2} \leq \frac{\lambda_s}{d_X} \leq \frac{k}{2}-\frac{1}{2k}\,$,
  and hence $\,\displaystyle \frac{\left|\lambda_s\right|}{d_X} \leq \left|\frac{k}{2}-\frac{1}{2k}\right| = \left|\frac{k^2-1}{2k}\right|$.
  Thus, taking a maximum over all pairwise points and rearranging, we obtain
  \[
    \norm{\Delta_s} \leq \left|\frac{k^2-1}{2k}\right|\norm{D_X}.
  \]
  The result is obtained by the transitivity of our minimum $\Delta$.
\end{proof}
\begin{corollary}
  \label{cor:Hom:Lip_dN}
  If $f:X\rightarrow Y$ is $k$-biLipschitz, then $\displaystyle \D{(X,Y)} \leq \left|\frac{k^2-1}{k}\right|\frac{\norm{D_X}}{\norm{D_Y}}$\,.
\end{corollary}
\begin{proof}
  The result follows from applying Theorem  \ref{thm:Hom:Lip_opt_decomp} to Theorem \ref{thm:Hom:dN_stabil}.
\end{proof}

\begin{remark}
  \label{rem:altbLip}
  A slightly different definition of biLipschitz functions was used in the work of Mahabadi et al.~\cite{MaMaMaRa2018}.
  They defined the biLipschitz constant of a map $f : X \to Y$ as the minimum $D$ such that for some $\lambda > 0$ we have $\lambda \cdot d_X(x,y) \leq d_Y( f(x), f(y)) \leq \lambda D \cdot d_X(x,y)$ for every $x,y \in X$.
  If we were to work with this definition, arguments similar to those used in Theorem  \ref{thm:Hom:Lip_opt_decomp} will give us the following bound in place of the one in Corollary \ref{cor:Hom:Lip_dN}: 
  \[
    \D{(X,Y)} \leq \left|\lambda(D-1)\right|\frac{\norm{D_X}}{\norm{D_Y}} \, .
  \]
\end{remark}

\section{Computational Experiments} \label{sec:comp}

We present experiments on simulated and real data that highlight the increased effectiveness of using the normalized bottleneck distance, as opposed to regular bottleneck distance, on persistence diagrams in data analysis pipelines.
As a proof of concept, we first apply the nonlinear DR method UMAP to a simple data set in 3D that we generated.
We then apply UMAP to images from the MNIST database of digit images \cite{De2012}. 
In the second experiment, we consider a collection of point clouds sampled from three distinct classes of 3D models that we study using cluster analysis.

\subsection{UMAP on Generated Data} \label{ssec:UMAP}

Uniform manifold approximation and projection (UMAP) is a recently developed nonlinear dimension reduction technique that employs a Riemannian metric that is estimated on the input data \cite{McHeMe2020}.
As such, for choices of parameters that define neighborhoods of points appropriately, one might expect UMAP to preserve homology of the data
and to do so better than linear methods such as JL projection or even multidimensional scaling.
To test this assertion, we apply UMAP to $200$ uniformly sampled points along a saddle boundary in $\R^3$ with radius $100$ and height $100$.
We refer to this original data set as $X$.
Using 100 nearest neighbors in the UMAP setting (and other parameters set at their default values), we reduced this data down to $\R^2$.
We refer to the reduced dataset as $Y$.

\begin{figure}[ht!]
    \centering
    \includegraphics[width=.7\textwidth]{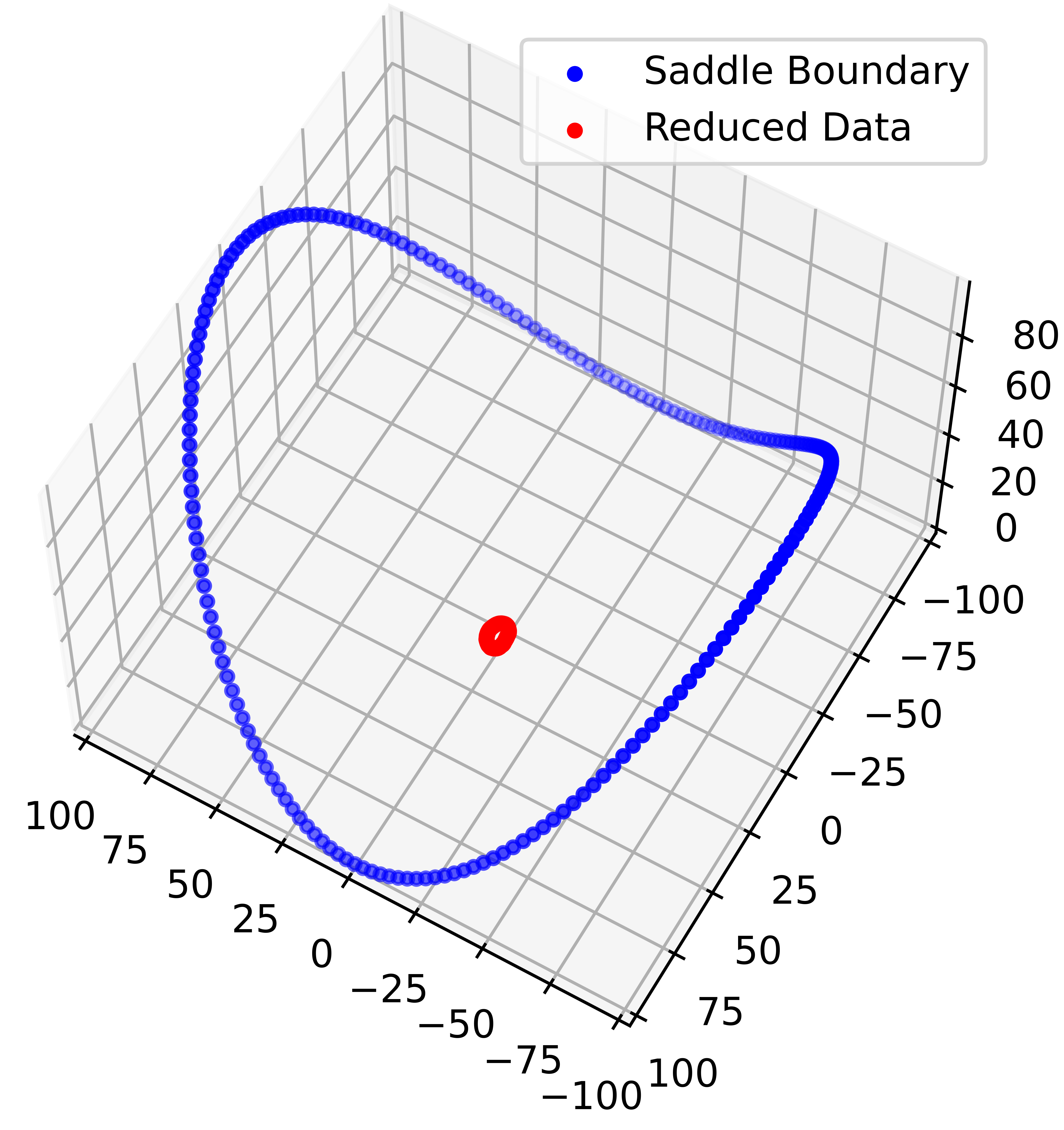}
    \caption{UMAP reduction of a saddle boundary in $\R^3$ to $\R^2$.}
    \label{fig:Hom:UMAP_Saddle}
\end{figure}

We can see in Figure \ref{fig:Hom:UMAP_Saddle} that UMAP does preserve the one-dimensional hole of the saddle boundary.
However, we can also clearly see that the data has been scaled down considerably.
This scaling shows up in the bottleneck distance---computing it directly, we find that $\dB(X,Y) \approx 2.2$.
We get the optimal scalar for the decomposition $\argmin_s \maxnorm{D_X - sD_Y}$ as $s \approx 26$.
Thus it is not surprising that $\dB$ performs poorly.

The normalized bottleneck performs better in this regard, with $\D(X,Y) \approx 0.012$.

\subsection{UMAP on MNIST Digit Images} \label{ssec:UMAPMNIST}

We study a subset of digit images from the MNIST database \cite{De2012} consisting of 1797 images.
Each image has the size of $8 \times 8$ pixels and is represented as a vector in $\R^{64}$.
We refer to this data set as $X$.
We reduce $X$ to $2$ dimensions using UMAP, and refer to the reduced data set as $Y$.
We obtain the default bottleneck distance between the two sets as $\dB(X,Y) \approx 4.34$, which is large due to the high dimensionality of $X$ and the scaling imposed by UMAP.
On the other hand, we obtain the normalized bottleneck distance as $\D(X,Y) \approx 0.056$, capturing the closeness of the two versions of the digits data set in a much better manner.

\subsection{Clustering of Point Clouds of Data using Persistence} \label{ssec:clust}

We study a set of point clouds sampled from the surfaces of three classes of objects: frogs, chairs, and the torus with varying scales.
We sample clouds of 1000 points each, selected randomly from the vertex sets of the meshes of four different frogs and four different chairs (see Figure \ref{fig:frogschairs}) from the Free3D database \cite{free3d}.
Each of these eight meshes has 10,000 vertices (hence, each data set represents a 10\% sample of the original vertex sets).
To this collection of eight data sets, we add four sets of 1000 points each, sampled uniformly at random from the surface of tori at varying scales (see bottom row in Figure \ref{fig:frogschairstoriPCDs}).

\begin{figure}[ht!]
  \includegraphics[scale=0.46]{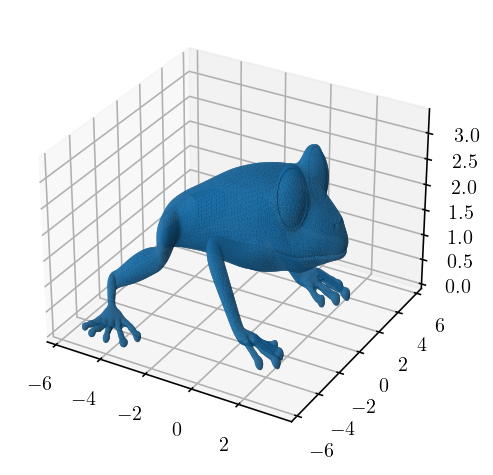}
  \includegraphics[scale=0.46]{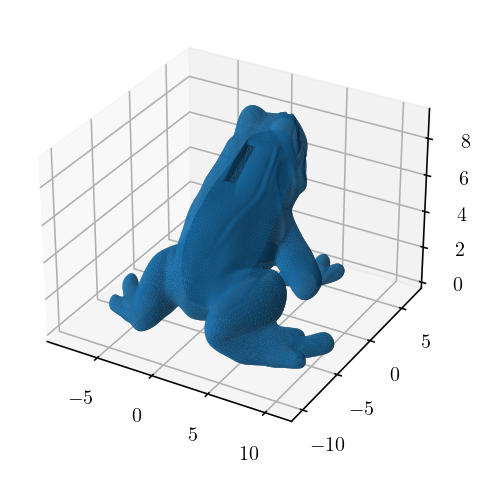}
  \includegraphics[scale=0.46]{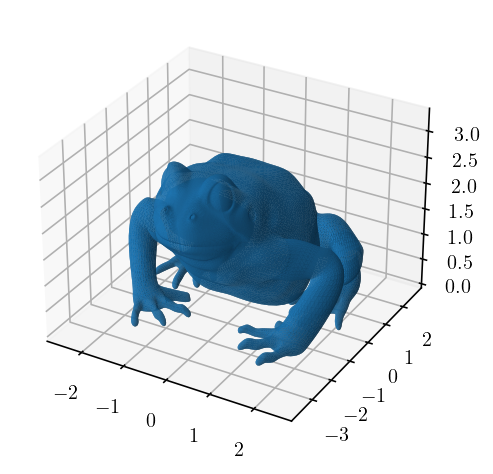}
  \includegraphics[scale=0.46]{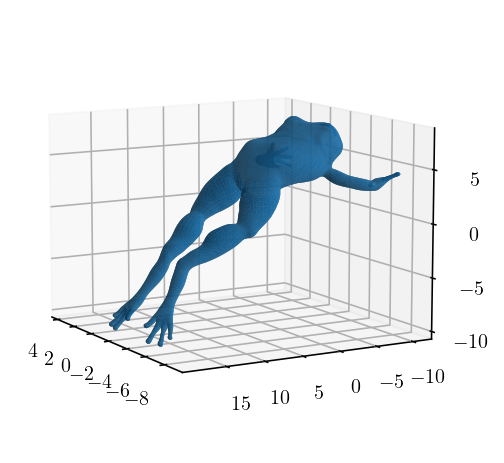}\\
  \includegraphics[scale=0.46]{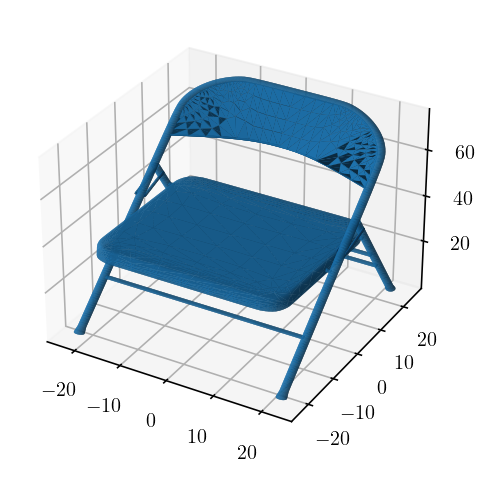}
  \includegraphics[scale=0.46]{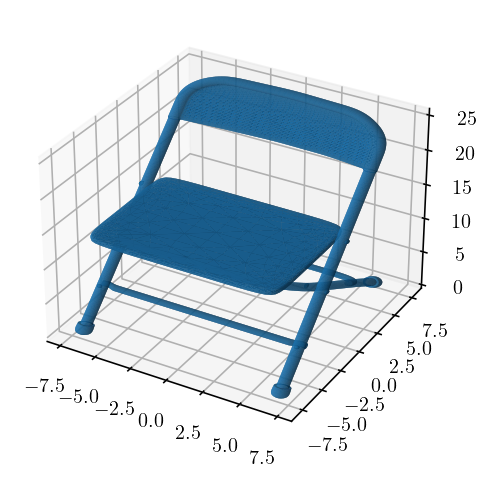}
  \includegraphics[scale=0.46]{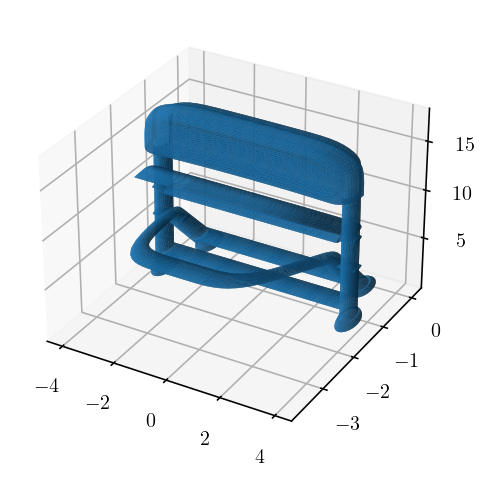}
  \includegraphics[scale=0.46]{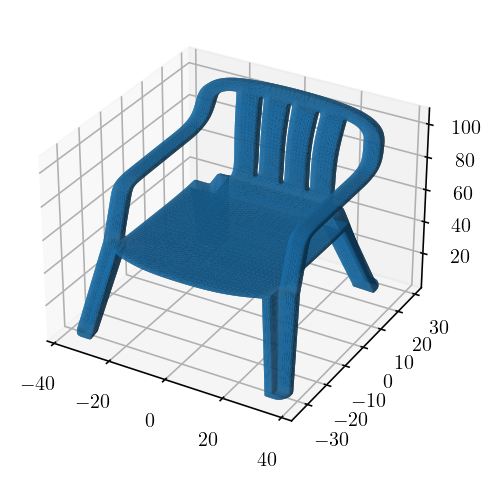}
  \caption{\label{fig:frogschairs}
    Meshes of frogs and chairs from the Free 3D \cite{free3d} data base.
    Note the varying scales across the models from each class.
    The frogs chosen are (from Left to Right): \href{https://free3d.com/3d-model/frog-v1--149240.html}{Frog V1 149420}, \href{https://free3d.com/3d-model/frog-v1--30593.html}{Frog V1 30593}, \href{https://free3d.com/3d-model/banjofrog-v1--699349.html}{Banjofrog V1 699349}, and \href{https://free3d.com/3d-model/frog-v1--64825.html}{Frog V1 64825}.
    The chairs chosen are (from Left to Right): \href{https://free3d.com/3d-model/-folding-chair-metal-v2--311510.html}{Folding Chair Metal V2 311510}, \href{https://free3d.com/3d-model/folding-chairs-v1--612720.html}{Folding Chair V1 612720}, \href{https://free3d.com/3d-model/fold-out-chair-folded-v1--372704.html}{Folded Fold Out Chair V1 372704}, and \href{https://free3d.com/3d-model/monobloc-chair-v1--691935.html}{Monobloc Chair V1 691935}.
  }
\end{figure}

\begin{figure}[ht!]
  \includegraphics[scale=0.46]{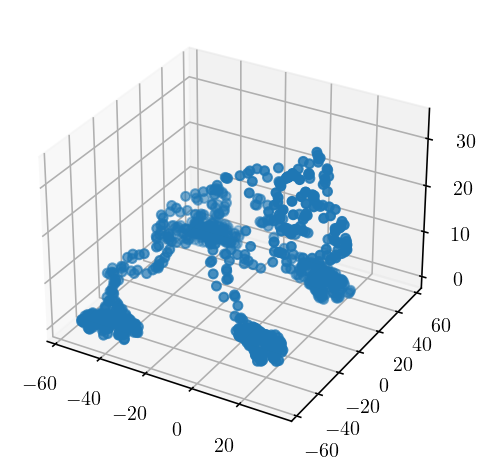}
  \includegraphics[scale=0.46]{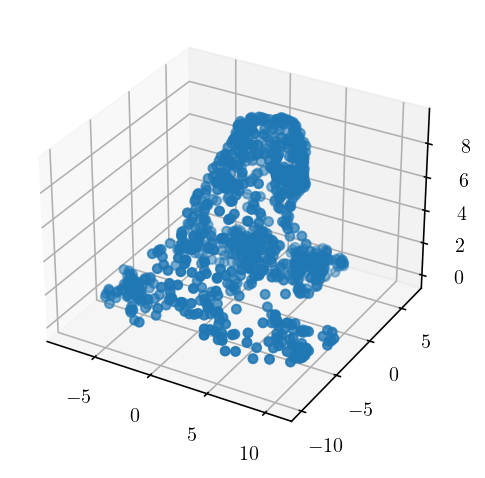}
  \includegraphics[scale=0.46]{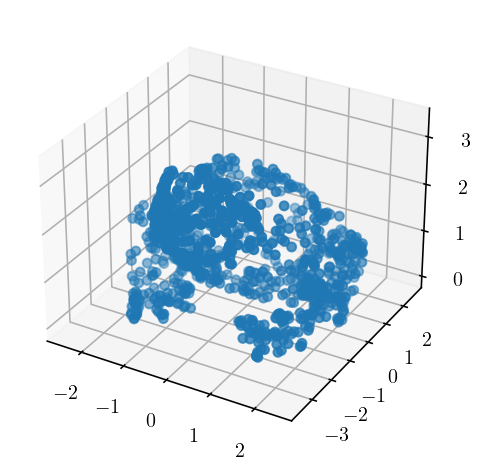}
  \includegraphics[scale=0.46]{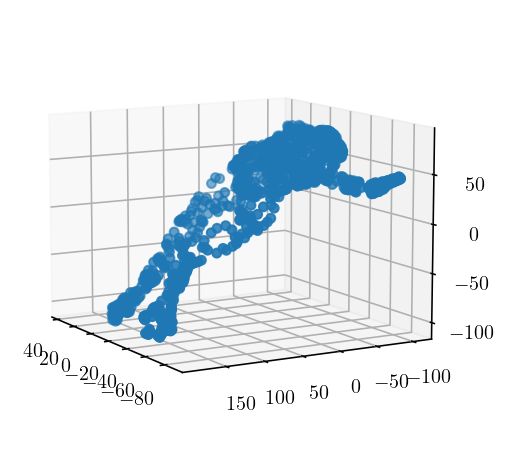}\\
  \includegraphics[scale=0.46]{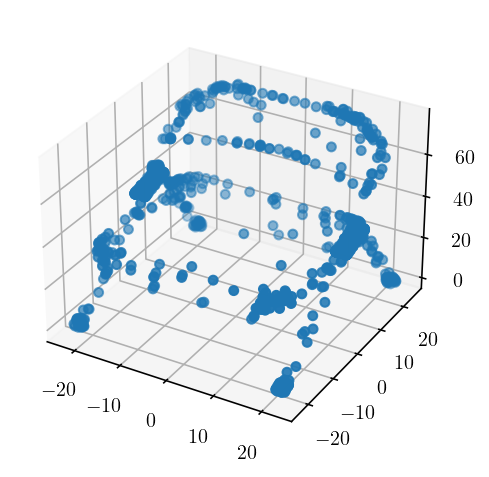}
  \includegraphics[scale=0.46]{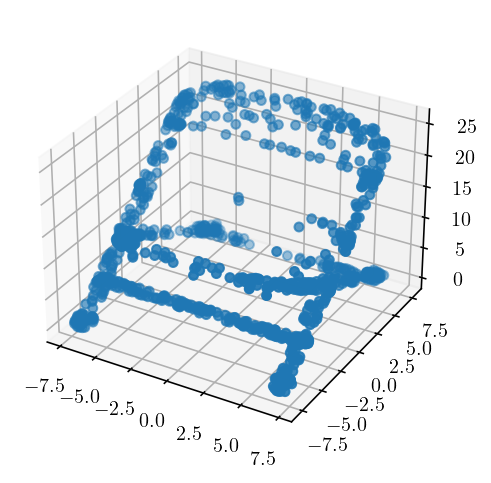}
  \includegraphics[scale=0.46]{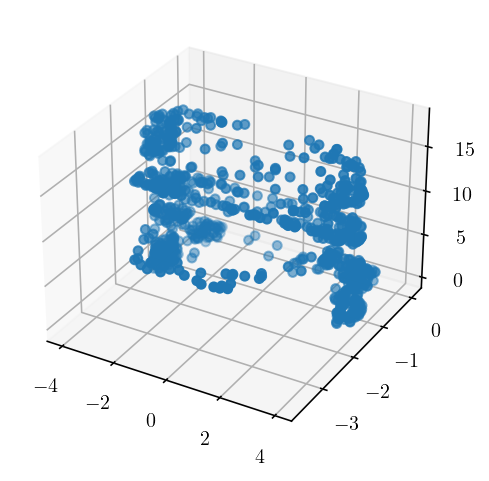}
  \includegraphics[scale=0.46]{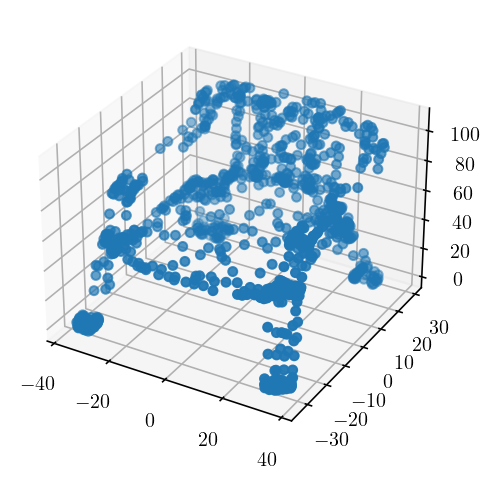}\\
  \includegraphics[scale=0.46]{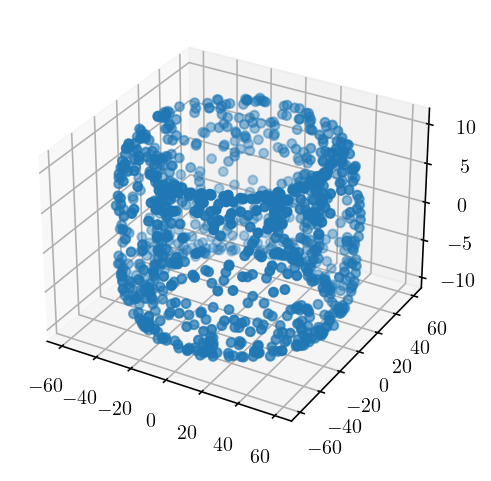}
  \includegraphics[scale=0.46]{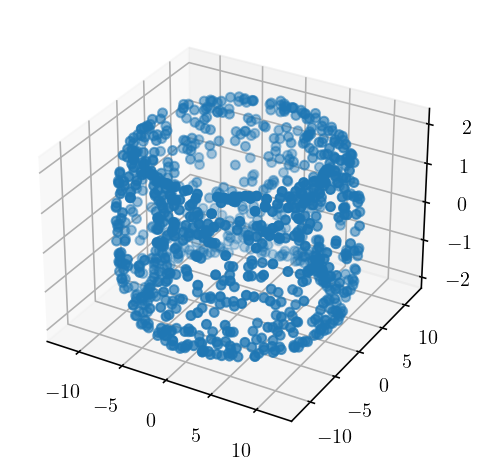}
  \includegraphics[scale=0.46]{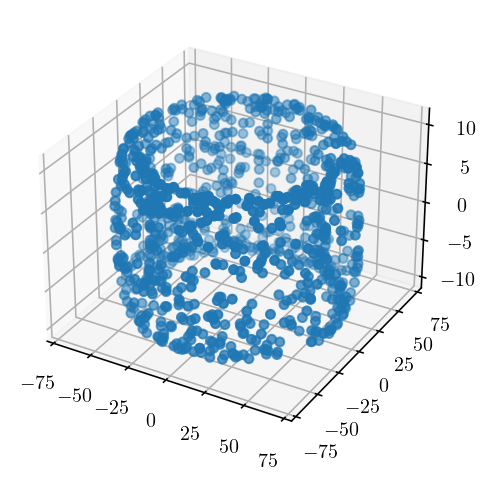}
  \includegraphics[scale=0.46]{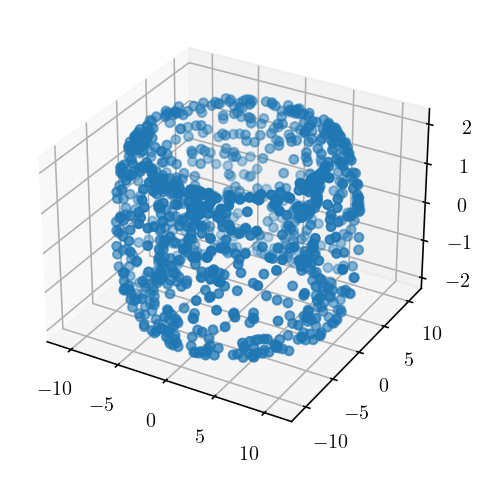}
  \caption{\label{fig:frogschairstoriPCDs}
    Point clouds sampled from the meshes of frogs and chairs shown in Figure \ref{fig:frogschairs} (top two rows) and point clouds sampled from surfaces of tori at varying scales (bottom row).
  }
\end{figure}

  We chose the point clouds from these spaces because the frogs, chairs, and tori have varying $H_1$ homology.
  The frogs are all $2$-spheres (so, $H_1$ homology groups are trivial), the tori have two significant $H_1$ features each, and the chairs each have multiple holes.
  We represent each point cloud by its $H_1$-persistence diagram and use distances between these PDs to cluster the set of point clouds.
  Given the distinct structures of the $H_1$ homology groups across the three classes of spaces, one expects to be able to separate them into three clusters.

  We first compute the regular bottleneck distances $\dB$ between the respective $H_1$ PDs and perform $k$-means clustering with $k=3$ using these distances.
  The result does not produce a clear separation of the three classes of objects (frogs, chairs, and tori).
  To visualize the result, we generated the multidimensional scaling (MDS)  \cite{CoCo2000} projection into 2D of the 12 data sets using these pairwise $\dB$ distances (see Figure \ref{fig:dBMDS}).
  The output from $3$-means clustering is clearly illustrated in this MDS projection---the two tori data sets with larger scales are grouped into isolated clusters of their own while all the remaining 10 data sets (all frogs, all chairs, and the two other tori) are grouped into a third cluster.

  On the other hand, the data sets are clearly separated into three disjoint clusters when the analysis is repeated with $\D$ distances between their $H_1$ PDs---see the MDS projection into 2D in Figure \ref{fig:dNMDS}.
  We can see that the four frog data sets are quite close to each other, which is as expected since they all have no nontrivial $H_1$ features.
  The four chair data sets appear to have higher relative separation than the four tori data sets, which also agrees with the observation that the chair models (in Figure \ref{fig:frogschairs}) all have multiple nontrivial $H_1$ features.

  This computational experiment clearly demonstrates that the normalized bottleneck distance can work as a much more accurate measure of distance when used in data analysis pipelines involving clustering.
  For similar reasons, $\D$ may be a better choice for distances between PDs to be used in machine learning frameworks.

\begin{figure}[ht!]
  \centering
  \includegraphics[scale=0.95]{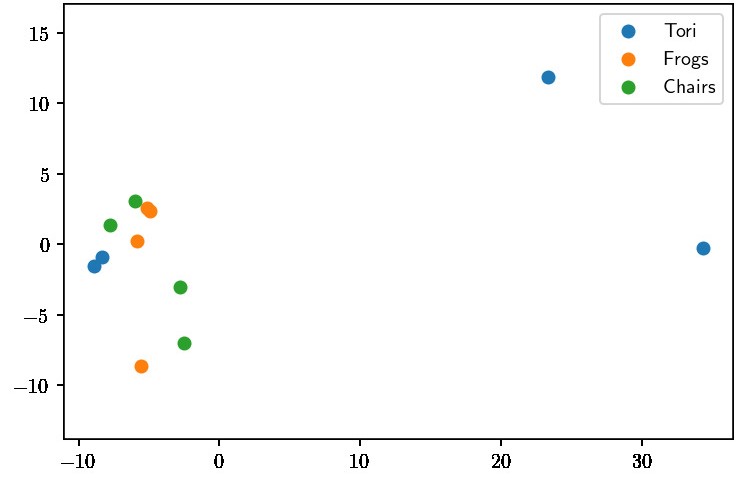}
    \caption{\label{fig:dBMDS}
      MDS projection into 2D of the 12 data sets using $\dB$ distances between their $H_1$ PDs.
      The two tori points on the far right correspond to the first and third data sets in the third row of Figure \ref{fig:frogschairstoriPCDs}, which have bigger scales compared to the other tori data sets.
      $3$-means clustering puts these two tori into isolated clusters of their own with the remaining 10 data sets grouped into the third cluster (as shown by the 10 points on the left).
    }
\end{figure}

\begin{figure}[hb!]
  \centering
  \includegraphics[scale=0.95]{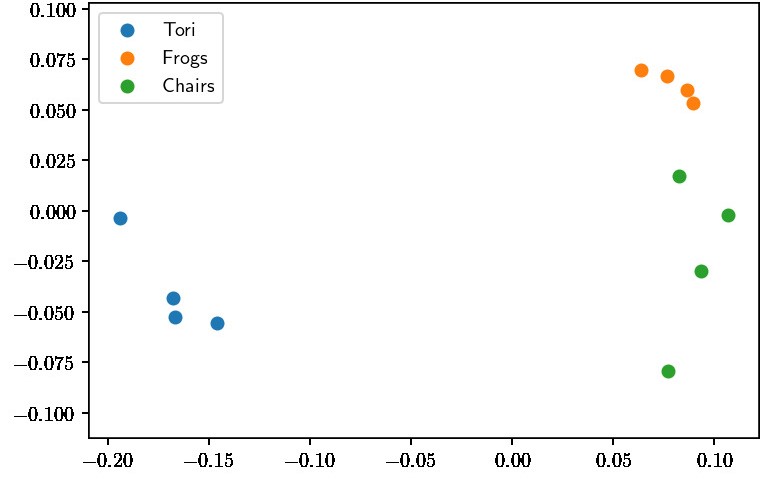}
  \caption{\label{fig:dNMDS}
    MDS projection into 2D of the 12 data sets using $\D$ distances between their $H_1$ PDs.
    The three classes of data sets (tori, frogs, and chairs) are clearly separated into clusters of their own with no overlaps.
  }
\end{figure}

\section{Discussion}

The normalized bottleneck $\D$ may indeed have advantages over the regular bottleneck distance when there is a large degree of scaling between two data sets.
It may further have advantages over the shift-invariant bottleneck distance $\op{d}_{\mathrm{sB}}$ \cite{CaKiSh2018} and the dilation-invariant bottleneck distance \cite{CaVlScKaMo2021} $\overline{\mathrm{d}_{\rm D}}$ (see Equation \ref{defn:Hom:DilBn}).
In practice, we do not have to compute optimal shifts/scalings in order to compute $\D$, which saves a significant amount of computational effort.
Furthermore, $\D$ comes equipped as a pseudometric, allowing values for $\D$ to be compared directly.

\smallskip
However, in some cases the optimal scaling value may itself be useful, in which case $\D$ may not prove more useful.
Furthermore, computing the optimal scaling value may indeed be necessary to obtain a tight bound from Theorem \ref{thm:Hom:dN_stabil}, if a bound may not already be known for $\maxnorm{\Delta}$.
Also, without an explicit optimal scaling value, it is unclear how much relative scaling is removed as opposed to how much scaling is being contributed from the normalization.
Further work is needed to establish a relationship between $\D(X,Y)$ and $\displaystyle \frac{\dB(X,Y)}{\max \{\diam(X),\diam(Y)\}}\,$.

\smallskip
Scaling the metric spaces by their respective diameters to define the normalized bottleneck distance (Definition \ref{defn:Hom:dN}) is arguably the natural choice.
Alternatively, scaling by other meaningful quantities capturing the scale of the metric space, e.g., eccentricity, might yield results that are less sensitive to the largest deviation than the diameter.
At the same time, deriving stability results corresponding to that for $\D$ (in Theorem \ref{thm:Hom:dN_stabil}) could be more challenging when the scaling is done using eccentricity.

\smallskip
For its application\add{s} is dimension reduction, the normalized bottleneck may be examined for a broader range of techniques.
Modern nonlinear reduction techniques such as UMAP \cite{McHeMe2020} and t-SNE \cite{vdMaHi2008} may be of particular interest.
In particular, while we demonstrated the effectiveness of $\D$ in practice for UMAP (see Sections \ref{ssec:UMAP} and \ref{ssec:UMAPMNIST}), it is highly desirable to derive a bound on $\D$ for UMAP similar to the one for mMDS in Corollary \ref{cor:Hom:mMDS_dN}.
Could we derive an explicit bound for $\maxnorm{\Delta}$ for UMAP?


\end{document}